\documentclass[12pt,letterpaper]{article}

\usepackage{fixltx2e}
\usepackage{textcomp}
\usepackage{fullpage}
\usepackage{amsfonts}
\usepackage{verbatim}
\usepackage[english]{babel}
\usepackage{pifont}
\usepackage{color}
\usepackage{setspace}
\usepackage{lscape}
\usepackage{indentfirst}
\usepackage[normalem]{ulem}
\usepackage{booktabs}
\usepackage{natbib}
\usepackage{float}
\usepackage{latexsym}
\usepackage{url}
\usepackage{hyperref}
\usepackage{epsfig}
\usepackage{graphicx}
\usepackage{amssymb}
\usepackage{amsmath}
\usepackage{bm}
\usepackage{array}
\usepackage[version=3]{mhchem}
\usepackage{ifthen}
\usepackage{caption}
\usepackage{hyperref}
\usepackage{amsthm}
\usepackage{amstext}

\usepackage{tikz}
\usetikzlibrary{matrix,arrows,positioning,automata,shapes,shapes.multipart}

\newtheorem{theorem}{Theorem}

\newtheorem{corollary}{Corollary}
\newtheorem{proposition}{Proposition}

\newcommand{\A}{{\mathcal A}}

\renewcommand{\section}[1]{%
\bigskip
\begin{center}
\begin{Large}
\normalfont\scshape #1
\medskip
\end{Large}
\end{center}}

\renewcommand{\subsection}[1]{%
\bigskip
\begin{center}
\begin{large}
\normalfont\itshape #1
\end{large}
\end{center}}

\renewcommand{\subsubsection}[1]{%
\vspace{2ex}
\noindent
\textit{#1.}---}

\renewcommand{\tableofcontents}{}

\bibpunct{(}{)}{;}{a}{}{,}  % this is a citation format command for natbib

\begin{document}

\begin{center}

\noindent{\Large \bf Which phylogenetic networks are merely trees with additional arcs?}
\bigskip

\noindent {\normalsize \sc Andrew R. Francis$^1$ and Mike Steel$^2$}\\
\noindent {\small \it 
$^1$Centre for Research in Mathematics, School of Computing, Engineering and Mathematics, University of Western Sydney, Australia;\\
$^2$Biomathematics Research Centre, Allan Wilson Centre, University of Canterbury, New Zealand}\\
\end{center}
\medskip
\noindent{\bf Corresponding author:} Mike Steel, Biomathematics Research Centre, University of Canterbury, Christchurch, 8041, Christchurch  E-mail: mike.steel@canterbury.ac.nz\\

% \vspace{1in}

\subsubsection{Abstract}A binary phylogenetic network may or may not be obtainable from a tree by the addition of directed edges (arcs) between tree arcs. Here,  we establish a precise and easily tested criterion (based on `2-SAT') that efficiently determines whether or not any given network can be realized in this way. Moreover, the proof provides a polynomial-time algorithm for finding one or more trees (when they exist) on which the network can be based. A number of interesting consequences are presented as corollaries;  these lead to some further relevant questions and observations, which we outline in the conclusion. \\
\noindent (Keywords: Phylogenetic network, 2-SAT, antichain, phylogenetic tree, algorithm, reticulate evolution)\\

\section{Introduction}

Starting from any rooted binary phylogenetic tree, if we sequentially add one or more arcs (directed edges), each placed from a point on one tree arc to a point on another tree arc, then provided no directed cycles  arise, we obtain a rooted binary phylogenetic network. Many classes of phylogenetic networks can be generated in this way, even if, at first, their descriptions seem somewhat different. For instance,  hybridization networks  are usually produced by adding two arcs from points on  tree arcs to meet at a new hybridization vertex, with a further arc leading to a hybrid offspring; however,  an equivalent network can be produced by starting with a phylogenetic tree and simply adding arcs just between tree arcs. 

Here, we explore a key observation due to Leo van Iersel~\citep{leo}, namely that not every binary phylogenetic network can be obtained from a tree by simply adding arcs between tree arcs.   In this paper, we provide a precise mathematical characterization of the networks that can be obtained in this manner. This in turn  allows us to readily show  that certain classes of networks are tree-based, while others are not. 
We then describe an efficient algorithm for determining whether any given network is tree-based and finding possible trees from which to build the network. We illustrate the use of this algorithm on a recent phylogenetic network concerning the complex hybrid evolution of wheat.

Informally, we say that a binary phylogenetic network is a `tree-based network'  if it can be obtained from a rooted binary phylogenetic tree by sequentially attaching arcs  between the arcs of the tree.  This concept  is relevant to the question of whether phylogenetic networks can be viewed as really just  trees with some reticulate arcs between the branches or whether some networks are inherently less tree-like, so that the concept of an `underlying tree' may be meaningless. This is particularly relevant to the ongoing debate about  whether the evolution of certain groups (e.g. prokaryotes) should be viewed as tree-like with reticulation or whether the very notion of a tree should be dispensed with \citep{dag, doo, mar}.  A network that is not tree-based cannot be described as tree-like evolution with 
directed links between the branches of the tree (at least for the taxa under study -- the existence of unsampled or extinct taxa \citep{szo13} can alter this conclusion, as we show). Conversely, a network that is tree-based can still allow for genuine reticulation events such as the formation of hybrid taxa from two ancestral lineages.

Phylogenetic networks  can be viewed as providing either an  `explicit' picture of reticulate evolution or as giving an `implicit' representations of conflict in the data ({\em c.f.} \cite{huson2010phylogenetic}, p.71). In the explicit setting, vertices having two incoming arcs correspond to hypothesised reticulate evolutionary events  such as hybrid evolution, endosymbiosis, and  lateral gene transfer  (either individual transfers, or  `highways' of lateral gene transfers \citep{ban13}).  In the  `implicit' setting, the networks are frequently unrooted, as in the popular `NeighborNet' method \citep{bry}, and the degree of reticulation is a measure of the extend to which trees constructed from different loci (`gene trees') disagree with each other, even though the evolution of the taxa may be  essentially tree-like (such networks can also help identify true reticulation when it is present \citep{hol08}).   

Conflicts between gene trees arise by well-studied random processes at the interface of population genetics and molecular evolution, such as incomplete lineage sorting, gene duplication and loss, and lateral gene transfer (see \cite{kno10} or \cite{szo15}).  In this case, there is is often assumed to be a `species tree', with non-reticulate processes (incomplete lineages sorting and gene duplication) occurring within the branches of the tree, and with the reticulate process of lateral gene transfer providing linking arcs between the branches. Provided the level of random lateral transfers is not too high it is still possible to infer a `central tendency' species tree accurately \citep{roc13, ste13}, as well as  correcting conflicting gene trees \citep{ban15}.  

 In this paper,  we are concerned with a more basic question arising for explicit phylogenetic networks  -- namely, if one has a rooted binary phylogenetic network, regardless of how this may have been obtained,  then we wish to determine whether or not it can be described as a tree with additional arcs.   As we discuss further in the conclusion, the interpretation of tree-based requires some care, as there may be other trees that equally well represent the network (so any given tree need not be a `central tendency' species tree).

The structure of this paper is as follows. We first provide a precise definition of the concept of a tree-based network, and then state our main result. After deriving a number of consequences from this result, we then show how it leads to a simple algorithm to test if a network is tree-based, and we provide a sample application.
We then discuss the delicate relationship between a network being tree-based and displaying a tree, before concluding with a number of observations and some questions.

\subsection{Definitions}
First, we make the definition of a `tree-based network' more precise. Given a set $X$ of taxa, a  \emph{binary phylogenetic network (over $X$)} refers to any directed acyclic graph $N = (V,A)$, for which:
\begin{itemize}
\item $X$ is the set of vertices that have out-degree 0 and in-degree 1 (leaves);
\item there is a unique vertex of in-degree 0, called the {\em root} (denoted $\rho$), which has out-degree 1 or 2;
\item every vertex other than $\rho$ or a leaf either has in-degree 2  and out-degree 1, or in-degree 1 and out-degree 2.
\end{itemize}

We  say that a binary phylogenetic network $N$ is a {\em tree-based network} (with base tree $T$) if
$N$ can be described as follows. First, subdivide each arc of $T$ as many times as required, and call the resulting degree-2 vertices {\em attachment points} and the resulting tree $T'$ a {\em support tree (for $N$ derived from $T$)}.  Next, sequentially place additional arcs between any two attachment points, provided that the network remains binary (i.e no  two additional arcs start or end at the same attachment point) and acyclic (i.e. no  directed cycle is created).  We call these additional arcs {\em linking arcs}. Any attachment point that is not incident with a linking arc is then suppressed.  Notice that this allows for parallel edges to be present in a tree-based network (if two attachment points are adjacent in $T'$ with a linking arc between them). 

Requiring a network $N$ to be based on a tree $T$ is a much stronger condition than just requiring that $N$  `displays' $T$.  We will  explore  the relationship between
these two concepts further in a later section. The interested reader is referred to ~\cite{huson2010phylogenetic} for general background on  phylogenetic networks.
\bigskip

Some basic observations to note at this point are as follows:
\begin{itemize}  
\item[(i)] All vertices of any network that is  based on $T$  are vertices of the support tree $T'$ (i.e. no new vertices are created, since a linking arc is not allowed to start   or end on another linking arc).
\item[(ii)] The order in which the additional arcs are attached in converting $T'$ to $N$ is not important.
\item[(iii)] A tree-based network can have different possible base trees; for example,
Fig.~\ref{fig:3tbones} shows  a binary network on $\{a,b,c\}$ that can be based on all three of the possible 3-taxon trees.	
\item[(iv)]   
Not all binary phylogenetic networks are tree-based, one example (from \cite{leo}) is shown in Fig.~\ref{fig:egs}(i) and another in Fig.~\ref{fig:egs}(iii).  
\end{itemize}

%% FIGURE ONE
%%%%%%%%%%%%%%%%%%%%%%%%%%
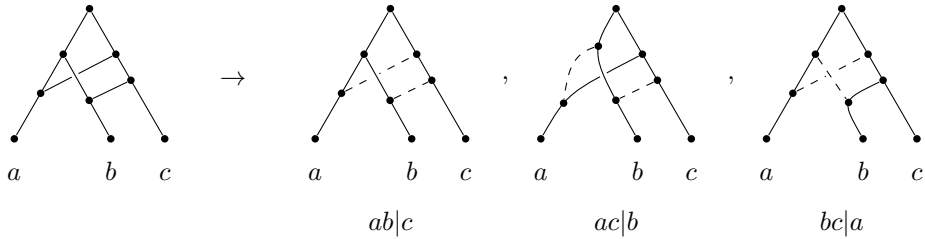
\begin{figure}
	\begin{center}
\begin{tikzpicture}
\footnotesize
\node[circle,fill,radius=2pt,inner sep=1pt](root){};
\path (root) ++(-120:2cm)   node[circle,fill,radius=2pt,inner sep=1pt] (a) {};
\path (root) ++(-120:.7cm)  node[circle,fill,radius=2pt,inner sep=1pt] (a1) {};
\path (root) ++(-120:1.3cm) node[circle,fill,radius=2pt,inner sep=1pt] (a2){}; 
\path (root) ++(-60:2cm)    node[circle,fill,radius=2pt,inner sep=1pt]  (c) {};
\path (root) ++(-60:.7cm)   node[circle,fill,radius=2pt,inner sep=1pt] (c1) {};
\path (root) ++(-60:1.1cm)  node[circle,fill,radius=2pt,inner sep=1pt] (c2){}; 
\path (a)-- node[pos=.65,circle,fill,radius=2pt,inner sep=1pt](b){} (c);
\draw(a)--(root)--(c) (a2)--(c1) ;
\draw[white,line width=1mm] (a1)--(b);
\path (a1) --  node[pos=.55,circle,fill,radius=2pt,inner sep=1pt] (b1){} (b); 
\draw (a1)--(b) (b1)--(c2);
\node[below=1mm of a]{$\mathstrut a$};
\node[below=1mm of b]{$\mathstrut b$};
\node[below=1mm of c]{$\mathstrut c$};
\node[right=1cm of c2] {$\rightarrow$};

\begin{scope}[xshift=4cm]
\node[circle,fill,radius=2pt,inner sep=1pt](root){};
\path (root) ++(-120:2cm)   node[circle,fill,radius=2pt,inner sep=1pt] (a) {};
\path (root) ++(-120:.7cm)  node[circle,fill,radius=2pt,inner sep=1pt] (a1) {};
\path (root) ++(-120:1.3cm) node[circle,fill,radius=2pt,inner sep=1pt] (a2){}; 
\path (root) ++(-60:2cm)    node[circle,fill,radius=2pt,inner sep=1pt]  (c) {};
\path (root) ++(-60:.7cm)   node[circle,fill,radius=2pt,inner sep=1pt] (c1) {};
\path (root) ++(-60:1.1cm)  node[circle,fill,radius=2pt,inner sep=1pt] (c2){}; 
\path (a)-- node[pos=.65,circle,fill,radius=2pt,inner sep=1pt](b){} (c);
\draw(a)--(root)--(c) ;
\draw[dashed] (a2)--(c1);
\draw[white,line width=1mm] (a1)--(b);
\draw (a1)--(b);
\path (a1) --  node[pos=.55,circle,fill,radius=2pt,inner sep=1pt] (b1){} (b); 
\draw[dashed]  (b1)--(c2);
\node[below=1mm of a]{$\mathstrut a$};
\node[below=1mm of b]{$\mathstrut b$};
\node[below=1mm of c]{$\mathstrut c$};
\node[below=2.5cm of root]{$ab|c$};
\node[right=.75cm of c2] {,};
\end{scope}

\begin{scope}[xshift=7cm]
\node[circle,fill,radius=2pt,inner sep=1pt](root){};
\path (root) ++(-120:2cm)   node[circle,fill,radius=2pt,inner sep=1pt] (a) {};
\path (root) ++(-120:.7cm)  node[circle,radius=2pt,inner sep=0pt] (a1) {};
\path (root) ++(-120:1.3cm) node[circle,radius=2pt,inner sep=0pt] (a2){}; 
\path (root) ++(-60:2cm)    node[circle,fill,radius=2pt,inner sep=1pt]  (c) {};
\path (root) ++(-60:.7cm)   node[circle,fill,radius=2pt,inner sep=1pt] (c1) {};
\path (root) ++(-60:1.1cm)  node[circle,fill,radius=2pt,inner sep=1pt] (c2){}; 
\path (a)-- node[pos=.65,circle,fill,radius=2pt,inner sep=1pt](b){} (c);
\draw (root)--(c) ;
\draw (c1)..controls (a2).. node[pos=.65,circle,fill,radius=2pt,inner sep=1pt](a2'){} (a);
\draw[white,line width=1mm] (root)..controls (a1).. (b);
\draw (root)..controls (a1).. node[pos=.35,circle,fill,radius=2pt,inner sep=1pt](b1){} node[pos=.83,circle,fill,radius=2pt,inner sep=1pt](b2){}(b);
\draw[dashed] (c2)--(b2);
\node[left=1mm of a1](a1v){};
\node[above=1mm of a2](a2v){};
\draw[dashed] (b1)..controls (a1v) and (a2v) .. (a2');\node[left=1mm of a1](a1v){};
\node[below=1mm of a]{$\mathstrut a$};
\node[below=1mm of b]{$\mathstrut b$};
\node[below=1mm of c]{$\mathstrut c$};
\node[below=2.5cm of root]{$ac|b$};
\node[right=.75cm of c2] {,};
\end{scope}

\begin{scope}[xshift=10cm]
\node[circle,fill,radius=2pt,inner sep=1pt](root){};
\path (root) ++(-120:2cm)   node[circle,fill,radius=2pt,inner sep=1pt] (a) {};
\path (root) ++(-120:.7cm)  node[circle,fill,radius=2pt,inner sep=1pt] (a1) {};
\path (root) ++(-120:1.3cm) node[circle,fill,radius=2pt,inner sep=1pt] (a2){}; 
\path (root) ++(-60:2cm)    node[circle,fill,radius=2pt,inner sep=1pt]  (c) {};
\path (root) ++(-60:.7cm)   node[circle,fill,radius=2pt,inner sep=1pt] (c1) {};
\path (root) ++(-60:1.1cm)  node[circle,fill,radius=2pt,inner sep=1pt] (c2){}; 
\path (a)-- node[pos=.65,circle,fill,radius=2pt,inner sep=1pt](b){} (c);
\draw(a)--(root)--(c) ;
\path (a1) --  node[pos=.55,circle,radius=2pt,inner sep=0pt] (b1){} (b); 
\draw(c2)..controls(b1)..node[pos=.5,circle,fill,radius=2pt,inner sep=1pt] (b1'){}(b);
\draw[dashed] (a2)--(c1);
\draw[white,line width=1mm] (a1)--(b1');
\draw[dashed] (a1)--(b1');
\node[below=1mm of a]{$\mathstrut a$};
\node[below=1mm of b]{$\mathstrut b$};
\node[below=1mm of c]{$\mathstrut c$};
\node[below=2.5cm of root]{$bc|a$};
\end{scope}

\end{tikzpicture}
	
	 \caption{A tree-based network on three leaves in which all possible trees on three leaves could be the base.}
	 \label{fig:3tbones}
	 \end{center}
\end{figure}

In contrast to this last point, several classes of networks are tree-based.  Clearly, horizontal gene transfer networks define one such class (since there is a canonical tree associated with each such network which contains every vertex of the network \citep{fs2015tree})  but so, too, are tree-child networks, as noted
by \cite{leo} (see Corollary~\ref{treechild} below).  Since any hybridization network is also a tree-child network,  it follows that every hybridization network is tree-based (as noted above).   
Our goal here is to characterize when a binary phylogenetic network is tree-based, and provide criteria for deciding whether a given network is tree-based, along with an algorithm to determine this. We also explore the subtle relationship between a network being based on a tree, and the weaker but more widely-known notion of the network displaying a tree. 

\begin{figure}
\begin{center}
\begin{tikzpicture}\footnotesize
\node[circle,fill,radius=2pt,inner sep=1pt](root){};
\node[below left=2cm of root,circle,fill,radius=2pt,inner sep=1pt](u){};
\node[below right=2cm of root,circle,fill,radius=2pt,inner sep=1pt](w){};
\node[below left=1.2cm of root,circle,fill,radius=2pt,inner sep=1pt](a1){};
\node[below right=.7cm of a1,circle,fill,radius=2pt,inner sep=1pt](v){};
\draw(u)--(root)--(w) (a1)--(v);
\node[below=1cm of u,circle,fill,radius=2pt,inner sep=1pt](b1){};
\node[below=1cm of v,circle,fill,radius=2pt,inner sep=1pt](b2){};
\draw(u)--(b1) (v)--(b2) (u)--(b2);
\draw[line width=5pt,white] (v)--(b1);
\draw (v)--(b1);

\node[below=1cm of b1,circle,fill,radius=2pt,inner sep=1pt](c1){};
\node[below=1cm of b2,circle,fill,radius=2pt,inner sep=1pt](c2){};
\node[below=1cm of c1,circle,fill,radius=2pt,inner sep=1pt](d1){}; 
\node[below=1cm of c2,circle,fill,radius=2pt,inner sep=1pt](d2){}; 

\draw (b1)--(d1) (b2)--(d2);
\draw[line width=5pt,white] (w) to[out=-100,in=45] (c1);
\draw (w) to[out=-100,in=45] (c1);
\draw (w) to[out=-90,in=45] (c2);

%labels:
\node[below=.5mm of d1] () {$\mathstrut a$};
\node[below=.5mm of d2] () {$\mathstrut b$};
\node[left=.5mm of u] () {$\mathstrut u$};
\node[right=.5mm of v] () {$\mathstrut v$};
\node[right=.5mm of w] () {$\mathstrut w$};
\node[below=1cm of d2] (i) {(i)};
\node[right=3cm of i] (ii) {(ii)};
\node[right=3cm of ii] (iii) {(iii)};

\begin{scope}[xshift=3.5cm,yshift=-2cm]
\node[circle,fill,radius=2pt,inner sep=1pt](root){};
\node[above=1mm of root](){$\rho$};
% \node[below left=2cm of root,circle,fill,radius=2pt,inner sep=1pt](u){};
% \node[below right=2cm of root,circle,fill,radius=2pt,inner sep=1pt](w){};
\path (root) ++(-120:2cm) node[circle,fill,radius=2pt,inner sep=1pt]  (u) {};
\path (root) ++(-60:2cm) node[circle,fill,radius=2pt,inner sep=1pt]  (w) {};

\path (root) ++(-120:.8cm) node[circle,fill,radius=2pt,inner sep=1pt]  (u1) {}; 
\path (root) ++(-120:.8cm) node[left](){$u$};
\path (root) ++(-120:1.5cm) node[circle,fill,radius=2pt,inner sep=1pt]  (u2) {};  %node[pos=.6,left](){$x_2$}
\path (root) ++(-120:1.5cm) node[left](){$v$};
\path (root) ++(-60:1cm) node[circle,fill,radius=2pt,inner sep=1pt]  (v1) {};  %node[pos=.5,right](){$x_4$}
\path (root) ++(-60:1cm) node[right](){$w$};
% \node[below left=1cm of root,circle,fill,radius=2pt,inner sep=1pt](u1){};
% \node[below left=1.5cm of root,circle,fill,radius=2pt,inner sep=1pt](u2){};
% \node[below right=1cm of root,circle,fill,radius=2pt,inner sep=1pt](v1){};
\draw(u)--(root)--(w) (v1)--node[pos=.5,circle,fill,radius=2pt,inner sep=1pt](mid){} node[pos=.45,below](){$x$} (u2);
\draw(mid)--(u1);
\node[below=.5mm of u] () {$\mathstrut a$};
\node[below=.5mm of w] () {$\mathstrut b$};

\end{scope}

\begin{scope}[xshift=7cm,yshift=-1.5cm]
\node[circle,fill,radius=2pt,inner sep=1pt](root){};
\path (root) ++(-120:3cm) node[circle,fill,radius=2pt,inner sep=1pt]  (u) {};
\path (root) ++(-120:1cm) node[circle,fill,radius=2pt,inner sep=1pt]  (u1) {};
\path (root) ++(-120:1.7cm) node[circle,fill,radius=2pt,inner sep=1pt]  (u2) {};
\path (root) ++(-60:3cm) node[circle,fill,radius=2pt,inner sep=1pt]  (v) {};
\path (root) ++(-60:1cm) node[circle,fill,radius=2pt,inner sep=1pt]  (v1) {};
\path (root) ++(-60:1.7cm) node[circle,fill,radius=2pt,inner sep=1pt]  (v2) {};
\draw[white](u)--node[pos=.5,circle,fill,black,radius=2pt,inner sep=1pt](b){} (v);
\node[above=5mm of b,circle,fill,radius=2pt,inner sep=1pt] (b1){};

\draw(u)--(root)--(v);
\draw(b1)--(b);
\draw(u2)--node[pos=.5,circle,fill,black,radius=2pt,inner sep=1pt] (x){} (b1)--node[pos=.5,circle,fill,black,radius=2pt,inner sep=1pt] (y){} (v2);
\draw(v1)--node[pos=.4,circle,fill,black,radius=2pt,inner sep=1pt] (x1){} node[pos=.6,circle,fill,black,radius=2pt,inner sep=1pt] (x2){} (x);
\draw(u1)--(x2) (x1)--(y);
\node[below=.5mm of u] () {$\mathstrut a$};
\node[below=.5mm of v] () {$\mathstrut c$};
\node[below=.5mm of b] () {$\mathstrut b$};

\end{scope}
\end{tikzpicture}
\caption{Some pertinent examples of binary networks.  Example (i) is from \cite{leo}. Note that, despite first appearances, (ii) is tree-based via the tree arcs $(u,x)$ and $(x,v)$ (so that $(u,v)$ and $(w,x)$ are linking arcs).  Neither (i) nor (iii) are tree-based, as can easily be checked by Proposition~\ref{mainpro2}, since (i) has an antichain consisting of three vertices but only two leaves, while (iii) has an antichain of four vertices but only three leaves.}
\label{fig:egs}
\end{center}
\end{figure}
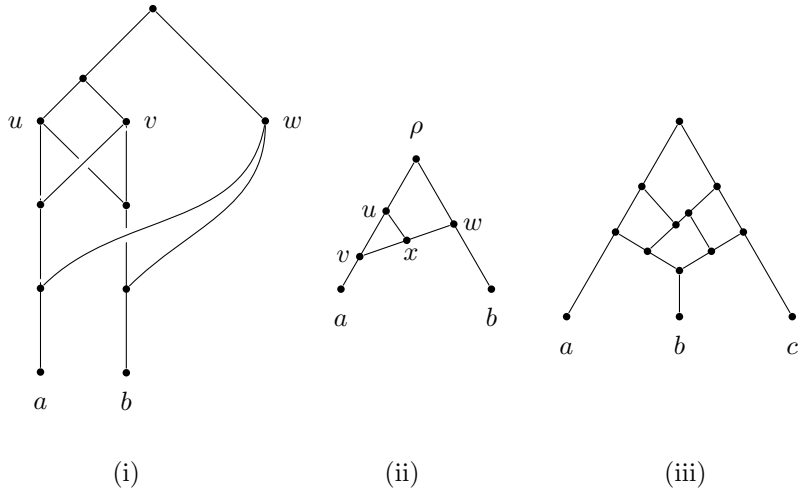

\section{The main theorem}

Our main theoretical result can be stated informally as follows. The question of whether or not a binary network is tree-based can be re-stated as an equivalent question in propositional logic called `2-SAT', and which is easily solved.  To make this precise we introduce some additional notation. 

Let $N= (V,A)$ be a rooted binary phylogenetic network on leaf set $X$.  
For an arc $a=(u,v) \in A$, 
we say that $u$ is the {\em source} and $v$ the {\em target} of $a$, and that $a$ is an {\em incoming arc} of $v$ and an {\em outgoing arc} of $v$.
Let $S_1$ be the subset of arcs in $A$ whose source has out-degree 1 or whose target has in-degree 1.
We say that a subset $S$ of $A$ is {\em admissible} if $S$ contains $S_1$ and satisfies the following two constraints for every $v \in V$:
\begin{itemize}
\item[($C_1$)] If $v$ has in-degree 2  then exactly one of its incoming arcs is in $S$.
\item[($C_2$)] If $v$ has out-degree 2 then at least one of its outgoing arcs is in $S$.
\end{itemize}

The problem 2-SAT is a classic and easily solved problem in logic to determine whether a conjunction of clauses each involving just two literals (or their negation) has a satisfying assignment.
For example,  suppose that,  in  a court case, witnesses have stated the following three opinions as to who may or may not  have been involved in a crime:  `{\em Peter or Susan}',  `{\em John or not Peter}', and `{\em not John or not Susan}'. As an instance of 2-SAT, the satisfiability question asks whether these three witness statements could all be correct.   In this case they can, for example, if John and Peter were involved in the crime but Susan was not.
With these concepts in hand, we can now state the main result, the proof of which is given in the Appendix. 

\begin{theorem}
\label{main}
\mbox{ } 
\begin{itemize}
\item[(a)]
A rooted binary phylogenetic network $N=(V,A)$  is tree-based if and only if there exists an admissible subset $S$ of $A$. In this case $S$ forms the arcs of a valid support tree for $N$ and the arcs in $A-S$ are
linking arcs. Moreover, there is a bijection between the set of admissible subsets $S$ of $A$ and the set of valid support trees for $N$.
\item[(b)] Determining whether $N$ is tree-based can be restated as a question of whether a particular instance of 2-SAT has a satisfying assignment, and this can be solved
in polynomial (linear) time. 
\end{itemize}
\end{theorem}

An immediate consequence of Part (a) of this theorem is that it \emph{is} the case that any non-tree-based rooted binary phylogenetic network can be expanded to become tree-based by the addition of extra arcs and leaves, as the next corollary shows. This is relevant in biology because these additional leaves may represent taxa that have become extinct in the past, and so can not be sampled today \citep{fou09, szo13}, or are still extant today but have not been included in the sample of taxa under study. In other words, any binary phylogenetic network can be realized as a tree with additional linking arcs, provided that one allows additional `unseen' taxa in the past to play a certain role in the evolution of the taxa sampled today.

\begin{corollary}
For any binary phylogenetic network $N$ over leaf set $X$ there exists a tree-based phylogenetic network $N^+$ over a leaf set $X^+$ that contains $X$ for which $N = N^+|X$.
\end{corollary}
Here $N^+|X$ is the restriction of $N^+$  to those vertices that have a path to at least one leaf in $X$ (it is obtained from $N^+$ by  deleting all vertices and arcs that do not lie on a path to a leaf in $X$, and then suppressing any vertices of in-degree and out-degree equal to 1).
\begin{proof}
Write $N=(V,A)$ and let $S=A$. Then $S$ fails to be admissible only by violations of condition ($C_1$). Thus we can convert $S$ into an admissible subset $S'$ by performing the following step 
for each vertex $v$ of $N$ that has in-degree 2.  Select either one of the incoming arcs arriving of $v$ -- say, $(u,v)$ -- and  subdivide this arc and attach a new leaf $x_v$ (specific to $v$) to the subdivided vertex $w$. Now 
remove $(u,v)$ from $S$ and replace it with the arcs $(u,w)$ and $(w, x_v)$.  Once this step is performed for all vertices of in-degree 2, the  set $S'$ of arcs of the resulting network $N'$ is admissible and so $N'$ is tree-based. Moreover, $N'|X = N$. 
\end{proof}

\subsection{Necessary and sufficient conditions for tree-based}

We now describe two further ways of characterizing tree-based. These are more immediate and of less direct algorithmic relevance, but we state them here as they provide a more complete picture of what `tree-based' means.   We say that a set of arcs in a directed graph is {\em independent} if no two arcs in the set share a vertex. The proof of the following result is provided in the Appendix.

\begin{proposition}
\label{mainpro}
Let $N= (V,A)$ be a rooted binary phylogenetic network on leaf set $X$.
The following are equivalent.
\begin{itemize}
\item[(a)]  $N$ is tree-based.
\item[(b)] There is an independent set $I$ of arcs of $N$
 for which $N' = (V, A-I)$ is a rooted tree.
 \item[(c)]  $N$ has a rooted spanning tree (with root $\rho$) that contains the arcs in $S_1$ and with all its leaves in $X$.
 \end{itemize}

\end{proposition}

Next we consider  a necessary condition for $N$ to be tree-based, based on the concept of `antichains'. This can provide a rapid way to verify that certain networks cannot be tree-based. An {\em antichain} in any directed graph is simply a subset $\A$ of vertices that has the property that there is no directed
path in the graph  from any one vertex in $\A$ to any other vertex in $\A$.   Let $N$ be a binary phylogenetic network. 
If, for any antichain $\A$ of non-leaf vertices in $N$, there exists at least $|\A|$ arc-disjoint paths from $\A$ to  the leaf set, we say it satisfies the {\em antichain-to-leaf property}. By a version of Menger's theorem for disjoint sets of vertices in directed graphs~\citep{boh01}, the antichain-to-leaf property is  equivalent to the  statement that for any  antichain $\A$ of non-leaf vertices in $N$,
at least $|\A|$ arcs of $N$ must be cut in order to separate $\A$ from $X$.

The following result, the proof of which is also in the Appendix,  provides a necessary condition for $N$ to be tree-based; if it fails we know immediately that $N$ cannot be tree-based. 

\begin{proposition}
\label{mainpro2}
If a binary phylogenetic network over leaf set $X$ is tree-based then it satisfies the antichain-to-leaf property.  In particular, the largest antichain in any tree-based network $N$ over $X$ has size exactly $|X|$.  Thus any tree-based network that has a larger antichain than the number of leaves cannot be tree-based.  
\end{proposition}

Proposition~\ref{mainpro2}   provides an easy way to verify that the network in Fig.~\ref{fig:egs}(i) is not tree-based, since it contains an antichain (the set $\{u, v, w\}$)  that is larger than the leaf set of the network.

It might  seem plausible that the antichain-to-leaf property is also a sufficient condition for a network to be tree-based.  Alas, this is not the case, and Fig.~\ref{fig:counter} shows a particular case where antichain-to-leaf property holds, yet the network is not tree-based.  

\begin{figure}[ht]
\begin{center}
\includegraphics[width=12cm]{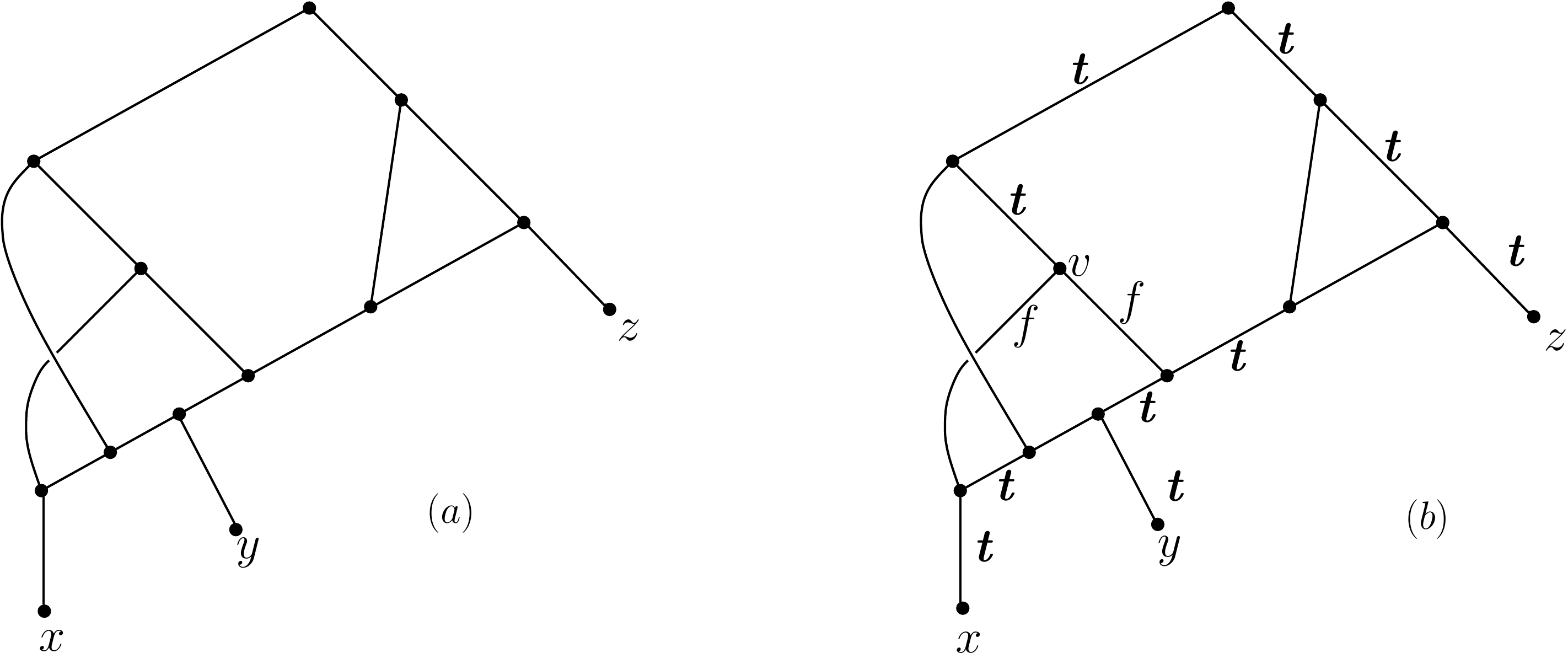}
\caption{(a) A network that is not tree-based, even though it satisfies the antichain-to-leaf property.   That this network fails to be tree based can be verified by applying Corollary~\ref{alg};  starting with the arcs in $S_1$  labelled $t$ (shown in bold in (b)) and applying the conditions ($C_1$)$'$ and ($C_2$)$'$ repeatedly, we are forced to label both of the arcs outgoing from $v$  by $f$, and at the next step ($C_2$)$'$ would assign one of these two arcs a second label $t$. }
\label{fig:counter}
\end{center}
\end{figure}

We turn now to some further necessary and sufficient conditions for $N$ to be tree-based. 
	
\begin{proposition}\label{specialcases}
Consider a binary phylogenetic network $N$ over leaf set $X$.
\begin{itemize}
\item[(i)]  If each vertex of $N$ of in-degree 2 has parents of out-degree 2, then $N$ is tree-based.
\item[(ii)]  If $N$ has a vertex of in-degree 2 whose parents both have out-degree 1, then $N$ is not tree-based.
\end{itemize}

\end{proposition}

\begin{proof}
Part (i) was established in the proof of Lemma 1 of \cite{gam} by an elegant application of Hall's matching theorem for digraphs.  For part (ii) note that the two parents form an antichain but paths from these parents to leaves both have to go through the edge below $v$, meaning they are not arc-disjoint, thereby violating Proposition~\ref{mainpro2}. 
\end{proof}

We end this section by showing how Theorem~\ref{main} provides a convenient way to verify that tree-child networks are tree-based, as are  tree-sibling 
networks (a result stated by~\cite{leo} without proof).

Recall that a network is a {\em tree-child} network if every non-leaf vertex is the parent of at least one vertex of in-degree 1, while (more generally) a {\em tree-sibling} network is a network for which every vertex $v$ of in-degree $2$ 
has a sibling $v'$ that has in-degree 1 (i.e. $v'$ is either a leaf of has 
outdegree 2). `Sibling' here means that the vertices share a parent. 
A network $N$ is {\em reticulation visible} if every vertex of $N$ of in-degree 2 has the property that for some leaf $x$ of $N$ all paths
from the root of $N$ to $x$ pass through $v$.  Tree-child networks are subset of the tree-sibling networks, but tree-sibling and reticulation visible represent different classes (and one is not a subset of the other).

\begin{corollary}
\label{treechild}
The class of tree-based networks includes tree-child networks (and thus hybridization networks), and, more generally,
tree-sibling networks.  It also includes the class of reticulation visible networks.
\end{corollary}

\begin{proof}
For tree-sibling networks, for each vertex $v$ of in-degree 2, select exactly one sibling $v'$ of $v$ that 
has in-degree 1, and if $p$ is the parent of
$v$ and $v'$ label the arc $(p,v)$ by $f$.  Then if $S$ is the set of arcs of $N$ minus the 
arcs labelled $f$ then $S$ is an admissible subset of arcs for $N$, and so, by Theorem \ref{mainpro}, $N$ is tree-based.
For reticulation-visible networks, \cite{gam}  showed that such networks satisfy the condition described above in part (i) of Proposition~\ref{specialcases}, which, 
in turn, they established suffices for $N$ to be tree-based.
\end{proof}

Note that although all tree-sibling networks are tree-based, it is easy to construct an example of a tree-based network that is not tree-sibling (an example is provided by \cite{leo}).

\section{An algorithm}

Theorem~\ref{main} furnishes a polynomial-time algorithm that takes any binary phylogenetic network $N$ and determines whether or not
it is tree-based. An extension can also then be used to determine a valid support tree for $N$, and indeed to compute all of these (however there may be exponentially many, and even counting the number is unlikely to be easy). 

First we describe a simple test that decides whether or not a network is tree-based, and which is  based on a well-known criteria for testing the satisfiability of any instance of 2-SAT  by taking the transitive closure of the implication relation (we give an example below).

To present this algorithm it is helpful to restate conditions ($C_1$) and ($C_2$) in an equivalent way, by making two modifications. Firstly, we will indicate that an arc $a$ is in $S$ or not in $S$ by assigning the arc the label $t$ (=`true') and $f$ (`false') respectively.  Secondly we will state the two conditions ($C_1$) and ($C_2$) in the form of implications (`if...then') to show how a label assigned to one arc can `force' the assignment of a label to an adjacent arc.

\begin{itemize}
\item[($C_1$)$'$]  when $v$ has in-degree 2, (i) if one of the in-coming arcs has label $t$ then the other in-coming is assigned label $f$, and 
(ii) if one of the in-coming arcs has label $f$ then the other in-coming arc is assigned label $t$.
\item[($C_2$)$'$]  when  $v$ has out-degree 2, if one of the out-going arcs has label $f$ then the other outgoing arc is assigned label $t$.
\end{itemize}

Now, let us label each arc in $S_1$ by  $t$, and then extend this labelling to other arcs by repeated applications rules  ($C_1$)$'$ and ($C_2$)$'$ when they apply.  It is clear that two things could happen: either a single label is assigned to (some or all of) the arcs of $N$ and the rules do not assign a label to any further arcs, or else at some point an arc could be assigned a label different from the one it has received earlier in the process.  
In turns out that $N$ is tree-based precisely if this latter case does not occur.  This is formalized in the following corollary of Theorem~\ref{main}, which is justified by  a well-known algorithm for testing satisfiability of 2-SAT \citep{kro}.

\begin{corollary}
\label{alg}
$N$ is tree-based if and only if case (i) does not arise under the following procedure:  Assign all arcs in $S_1$ label $t$ and then repeatedly apply conditions $(C_1)$$'$ and $(C_2)$$'$ to extend this labelling to other arcs of $N$, until either (i) an arc is assigned a label different from its existing label or (ii) the conditions can no longer be applied.  
\end{corollary}

Notice that the only arcs that do not receive an immediate label by their membership of $S_1$ are the pairs of arcs that are incoming to a vertex of in-degree 2. If the label for one of these arcs is subsequently determined (by application of the $C'$ conditions) then the status of the other arc in the pair  is fixed by ($C_1$)$'$.  
An example of how this algorithm works is provided in Fig~\ref{fig:counter}. In this case the algorithm detects that the network is not tree-based, since it leads to the case (i) where an arc is assigned a label different from its existing label. 
Although this algorithm is easy to apply by hand on small examples, for very large networks there exist faster (linear time) algorithms for deciding satisfiability of 2-SAT  and these  could be applied, however these are more technical to describe \citep{asp}.

Suppose now that each arc of $N$ receives at most one label. Then $N$ is tree-based, and if every arc of $N$ gets a label then, by Theorem~\ref{main}, there is unique support tree for $N$.  However another possibility is that only some of the arcs of $N$ are assigned a label.  In this case there exists more than one support tree (though the network may still only be based on one possible phylogenetic tree).  

To find a support tree for $N$ it suffices to select any arc $a$ that remains unlabelled at the end of the process described above, then assign one label  ($t$ or $f$)  to $a$  and apply the process again of extending the labelling using repeated applications of ($C_1$)$'$ and ($C_2$)$'$. 
We can then continue this procedure (selecting an unlabelled arc and extending the labelling so far obtained)  until all arcs receive a label. The arcs labelled $t$ then correspond to arcs of a  support tree for $N$, and the arcs labelled $f$ are the linking arcs. 

It is possible in this way to generate all the possible support trees for $N$, however there may be exponentially many of them, since if $N$ has $k$ vertices of in-degree 2, then the number of support trees can be as large as $2^k$.   Even counting the number of support trees may be hard, since counting the number of satisfying solutions of 2-SAT is known to be \#P-complete \citep{val}.

\subsection{Example}

We now provide a simple illustration of how this algorithm works by applying it to a phylogenetic network proposed recently by \cite{marc} to represent the complex hybrid evolution of bread wheat.   Our application here is not intended to provide support for or against particular claims in that paper. Rather, the purpose is to show how the algorithm can be applied to a small but realistic phylogenetic  network to determine whether or not it is  tree-based, and if it is, to illustrate how the tree(s) and linking arcs can be readily  identified.

\begin{figure}[ht]
\begin{center}
\resizebox{11cm}{!}{
\includegraphics{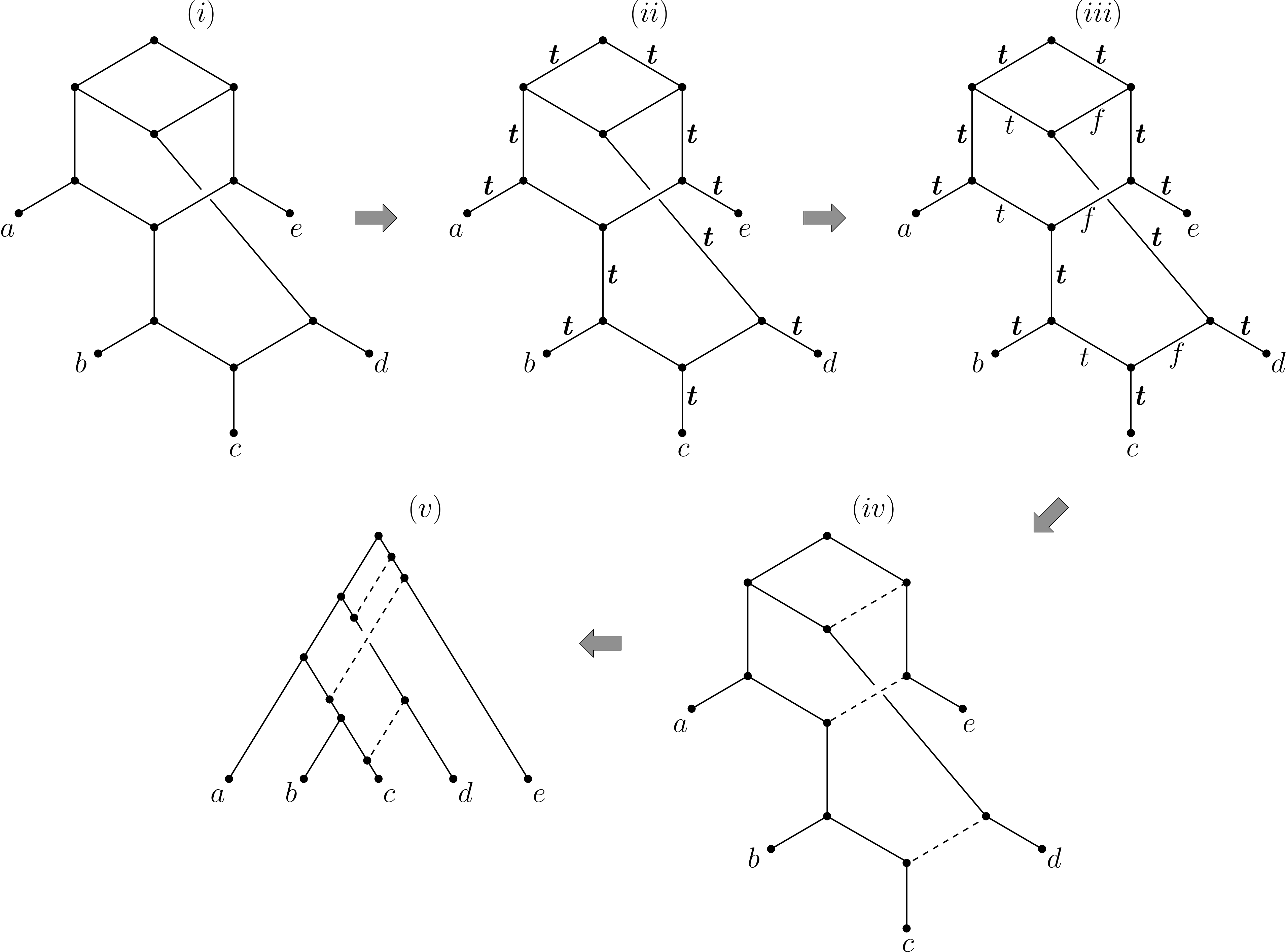}
}
\caption{ (i) A network from \cite{marc}  showing three ancient hybridization events in the evolution of bread wheat.  Corollary~\ref{alg} terminates at (ii) to show that the network is tree-based. Finding a particular support tree requires selecting a label for an unlabelled arc, extending the labelling using the rules ($C_1$)$'$ and ($C_2$)$'$ repeatedly, and continuing this process. In this example, the  six unlabelled arcs consist of three pairs, with the arcs in each pair incoming to one of the three vertices of in-degree 2. Assigning a label ($t$ or $f$) to an arc in one pair determines the assignment for the other arc in that pair but does not force any further arc assignments. Thus three independent assignments can be made for each pair, leading to $2^3 =8$ choices in total. For the particular choice shown in (iii) we obtain the support tree $T'$ shown in (iv) and thereby the tree-based representation shown in (v).
In this example, the seven rooted binary phylogenetic  trees that the network can be based on are all distinct.}
\label{wheat}
\end{center}
\end{figure}

Fig.~\ref{wheat}(i) shows a binary phylogenetic network on five leaves and three reticulations (vertices of in-degree 2). This network is essentially equivalent to the one shown in Fig. 3 of \cite{marc}, under the taxon labelling {\em a=Triticum uartu, b= Triticum turgidum, c= Triticum aestivum, d= Aegilops tauschii, e= Aegilops speltoides}.

First notice that the algorithm in Corollary~\ref{alg} tells us immediately that $N$ is tree-based, since the initial labelling of the arcs in $S_1$ by $t$ (shown in Fig.~\ref{wheat}(i)) 
does not extend further.
To find a support tree we see that assigning $t$ to either of the two arcs arriving at the lowest recitulation vertex does not cause any dual-labelled arc to arise. The same holds at the other two reticulation vertices, and these choices can all be made independently.  Thus there are $2^3=8$  eight possible support trees, and in this case the eight associated rooted binary phylogenetic trees (that the support trees are subdivisions of) are all distinct.

\section{The trees displayed by a tree-based network}

We have seen that a tree-based network can be based on more than one tree. An obvious question then is what one can say about the trees that
can act as a base for a given tree-based network $N$.  There is a related notion that applies to any binary phylogenetic network (tree-based or not), namely the 
concept of `displaying' a rooted phylogenetic tree, which we need to recall first.   Given a binary phylogenetic network $N$ over $X$, $N$ is said to {\em display} a rooted binary phylogenetic tree $T$ if $T$ can be obtained from $N$ by deleting arcs and vertices, and suppressing any resulting vertices of in-degree and out-degree equal to 1 \citep{cor}. Notice that if $N$ has at most $k$ vertices of in-degree 2, then it can display at most $2^k$ trees, and there has been some recent interest in identifying a class of networks for which this holds \citep{wil} or quantifying the extent to which it can fail \citep{cor}. A second active area of interest has involved determining  the computational complexity of deciding (for various categories of networks) whether a given network $N$ displays a given tree $T$ \citep{gam, leo2}.

It  is clear that if $N$ is a tree-based network, and is based on $T$, then $N$ must display $T$, since we can just delete the linking arcs, and suppress any resulting vertices of in-degree and out-degree equal to 1. However, it is possible for a tree-based network to display the tree $T$ but fail to be based on $T$ (see Fig.~\ref{fig:disp-not-based}). 
In other words, the  notion of a network being based on a tree is stronger than simply displaying the tree.

\begin{figure}[ht]
\begin{center}
\includegraphics[width=8cm]{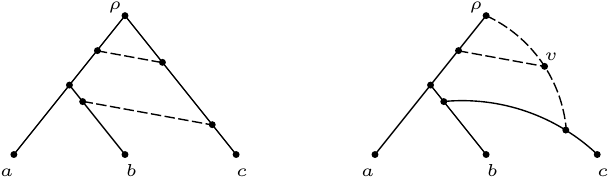}
\caption{A network that is based on a tree $ab|c$ (left) and displays $bc|a$ (right), but is not \emph{based} on the tree $bc|a$. Notice that in the right-hand network, vertex $v$ requires a linking arc to be attached to another linking arc.}
\label{fig:disp-not-based}
\end{center}
\end{figure}

Moreover, it turns out that the set of trees that are displayed by a  network that is based on $T$ need not bear any relation at all to $T$; indeed, for each positive integer $n$, there is a tree-based network displaying all trees on $n$ leaves.  This network can be based on \emph{any} tree on $n$ leaves as the following result shows (its proof is also in the Appendix):

\begin{proposition}\label{prop1}
For any $n\geq 2$ and any rooted binary phylogenetic $X$-tree $T$ on $n$ leaves, there is a tree-based binary network $N$ over $X$, based on $T$ and
with order $n^3$ linking arcs, such that $N$ displays {\em all} rooted binary phylogenetic $X$-trees.
\end{proposition}

\section{Concluding comments}

We end with some final comments.
\begin{enumerate}

\item
Establishing that a network is based on a tree $T$ does not necessarily mean that the evolution of the taxa under study was primarily represented by $T$ (or, indeed, on any rooted tree) with just some additional transfer events (like horizontal gene transfer, or endosymbiosis) between branches of the tree. As we have seen, hybridization networks can also be tree-based, even though they are described somewhat differently. Rather, tree-based means that one {\em can} represent evolution using a rooted tree and linking arcs, and this does not, in itself, confer  or require any particular mechanism of evolution for the taxa under study. 

	\item Our results suggest a number of further relevant questions. We have seen that a tree-based network $N$ can be based on more than one tree. However,   given a network, \emph{how many} base trees can it have?  
	\begin{enumerate}
		\item Is it possible to characterize the set of rooted binary phylogenetic trees on which $N$ can be based?  
		\item Given a tree-based network $N$ and an arbitrary rooted binary phylogenetic tree $T$, can it be decided in polynomial time whether or not $N$ is based on $T$?
		\item Is it possible that there is a network on a leaf set $X$ that is a tree-based network for \textbf{all} trees on $X$?  The answer is ``yes'' for $|X|=3$, as shown in Fig.~\ref{fig:3tbones}.

	\end{enumerate}

	\item The networks we have studied so far are required to be acyclic. However, basing a network $N$ on a tree $T$ suggests adopting a possibly stronger condition that relies on the assignment of an ordering of the vertices of $N$ to reflect  the temporal nature of vertical (tree-like) and horizontal (reticulate) evolution. More precisely, suppose that $N$ is a network based on $T$. A map $t$ from the vertices of $N$ to the real numbers (or the integers) is then a {\em strong temporal ordering} for $N$ relative to a valid support tree $T'$ derived from $T$ (i.e. one containing all the vertices of $N$), provided that $t$ satisfies the two properties:
	
	\begin{itemize}
	\item[(i)] If $(u,v)$ is any arc of $T'$, then $t(u)<t(v)$.
	\item[(ii)] If $(u,v)$ is a linking arc, then $t(u) = t(v)$.
	\end{itemize}
Condition (i) reflects the biological point that vertical evolution (a lineage persisting through time  plus speciation events) is proceeding with a natural time scale.  Condition (ii) captures the notion that reticulate evolution requires the two donor species to both be extant at some time in the past.
However, if we allow for additional species (not sampled, or perhaps now extinct, \citep{szo15}) to play a role in evolution, then Condition (ii) needs to be relaxed to the following condition:
	\begin{itemize}
	\item[(ii)$'$] If $(u,v)$ is a linking arc then $t(u) \leq t(v)$,
	\end{itemize}
When $N$ satisfies (i) and (ii)$'$, we say that $t$ is a {\em weak temporal ordering} for  $N$ relative to $T'$.

Notice that the (weak) temporal ordering condition in itself implies that $N$ must be acyclic, since if $v_1, \ldots, v_k=v_1$ is a directed cycle in a tree-based network, then
some pair of adjacent vertices in the cycle -- say $v_i$ and $v_{i+1}$ -- forms an arc of the support tree and so $t(v_i)< t(v_{i+1})$. However, since the $t$-values
of the vertices in the remainder of the path from $v_1$ to $v_k=v_1$  is non-decreasing (by (i) and (ii)$'$), this would imply that $t(v_1)<t(v_k) = t(v_1)$, which is a contradiction.
Fig.~\ref{fig:nontime} illustrates three tree-based networks that have no strong temporal ordering relative to any valid support tree.

It turns out that {\em every} acyclic network (and thus every tree-based network) has a weak temporal ordering.  To see this, note that because $N$ is an acyclic directed graph, it is possible to order the vertices $v_1, v_2, \ldots$ so that if $(v_i, v_j)$ is an arc of $N$ then
$i<j$ (Proposition 1.4.2 of \cite{ban}).   Thus if we let $t(v_i) =i$ for each $i$, we obtain a weak temporal ordering for $N$.  
In other words, if we accept the justification for relaxing temporal ordering based on the possible role of unsampled or extinct taxa in the reticulate evolution of the extant species under study, then the resulting weak temporal ordering constraint does not provide any real restriction on the class of tree-based networks.

%% FIGURE SIX
%%%%%%%%%%%%%%%%%%%%%%%%%%
\begin{figure}[ht]
\begin{center}
\begin{tikzpicture}

\node[circle,fill,radius=2pt,inner sep=1pt](root){};
\path (root) ++(-120:2cm)   node[circle,fill,radius=2pt,inner sep=1pt] (p) {};
\path (root) ++(-120:.8cm)  node[circle,fill,radius=2pt,inner sep=1pt] (a) {};
\path (root) ++(-120:1.3cm) node[circle,fill,radius=2pt,inner sep=1pt] (b){}; 
\path (root) ++(-60:2cm)    node[circle,fill,radius=2pt,inner sep=1pt]  (out) {};
\path (root) ++(-60:.8cm)   node[circle,fill,radius=2pt,inner sep=1pt] (d) {};
\path (root) ++(-60:1.3cm)  node[circle,fill,radius=2pt,inner sep=1pt] (c){}; 

\draw[->](b)--(d);
\draw[white, line width=1mm](a)--(c);
\draw[->](c)--(a);
\draw (p)--(root)--(out);
\node[below=1mm of p] () {$\mathstrut a$};
\node[below=1mm of out] () {$\mathstrut b$};
\node[below=3cm of root] () {$(i)$};

\begin{scope}[xshift=4cm]

% Part (ii)
% ========
\node[circle,fill,radius=2pt,inner sep=1pt](root){};
\path (root) ++(-120:2cm)   node[circle,fill,radius=2pt,inner sep=1pt] (p) {};
\path (root) ++(-120:.8cm)  node[circle,fill,radius=2pt,inner sep=1pt] (a) {};
\path (root) ++(-120:1.3cm) node[circle,fill,radius=2pt,inner sep=1pt] (b){}; 
\path (root) ++(-60:2cm)    node[circle,fill,radius=2pt,inner sep=1pt]  (out) {};
\path (root) ++(-60:.8cm)   node[circle,fill,radius=2pt,inner sep=1pt] (d) {};
\path (root) ++(-60:1.3cm)  node[circle,fill,radius=2pt,inner sep=1pt] (c){}; 

\draw[->](d)--(b);
\draw[white, line width=1mm](a)--(c);
\draw[->](a)--(c);
\draw (p)--(root)--(out);

\node[below=1mm of p] () {$\mathstrut a$};
\node[below=1mm of out] () {$\mathstrut b$};
\node[below=3cm of root] () {$(ii)$};

\end{scope}

\begin{scope}[xshift=8cm]

% Part (iii)
% ========
\node[circle,fill,radius=2pt,inner sep=1pt](root){};
\path (root) ++(-120:2cm)   node[circle,fill,radius=2pt,inner sep=1pt] (p) {};
\path (root) ++(-120:.8cm)  node[circle,fill,radius=2pt,inner sep=1pt] (a) {};
\path (root) ++(-120:.5cm) node[circle,fill,radius=2pt,inner sep=1pt] (b){}; 
\path (root) ++(-60:2cm)    node[circle,fill,radius=2pt,inner sep=1pt]  (out) {};
\path (root) ++(-60:1.6cm)   node[circle,fill,radius=2pt,inner sep=1pt] (d) {};
\path (root) ++(-60:1cm)  node[circle,fill,radius=2pt,inner sep=1pt] (c){}; 

\draw[->](b)--(d);
\draw[white, line width=1mm](a)--(c);
\draw[->](a)--(c);
\draw (p)--(root)--(out);

\node[below=1mm of p] () {$\mathstrut a$};
\node[below=1mm of out] () {$\mathstrut b$};
\node[below=3cm of root] () {$(iii)$};

\end{scope}
\end{tikzpicture}
\caption{Three networks with no strong temporal ordering relative to any valid support tree.
 Network (i) fails to be acyclic (and so is technically not even a binary phylogenetic network), but Networks (ii) and (iii) are acyclic (and so have a weak temporal ordering).}
\label{fig:nontime}
\end{center}
\end{figure}
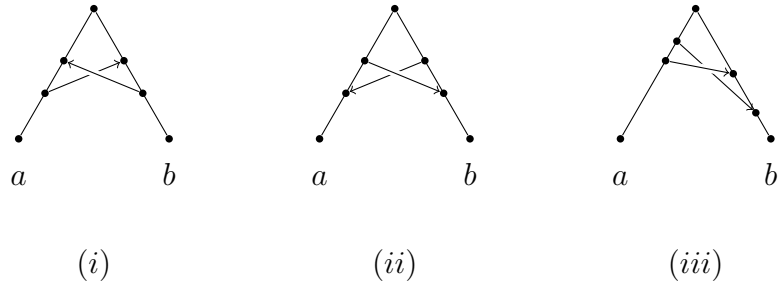

\end{enumerate}

\section{Funding} 
MS thanks the NZ Marsden Fund and the Allan Wilson Centre for helping fund this research.
ARF thanks the Australian Research Council via FT100100898 for funding this research.

\section{Acknowledgements}
We  thank Leo van Iersel for numerous helpful suggestions concerning this paper.

\newpage

\section{Appendix: Mathematical proofs}
\subsection{Proof of Theorem \ref{main}}

\begin{proof}
{\em Part (a)} Suppose that $N$ is tree-based, and let $T'$ be a support tree for $N$. Then the set $S$ of arcs of $T'$ contains $S_1$, and $S$ also satisfies conditions ($C_1$) and ($C_2$) for every vertex $v \in V$ of in-degree or out-degree 2, respectively.  Thus $S$ is admissible. 

Conversely, suppose that $S$ is an admissible subset of arcs of $N$.  Consider the network $N' = (V,S)$
consisting of all the vertices in $N$ and just the arcs in $S$.  We claim that this is a rooted tree, with root $\rho$ (the root of $N$) and leaf set $X$ (the leaf set of $N$). Firstly, notice that
$N'$ has no vertex of in-degree 2, by condition ($C_1$). Secondly,  every arc $a$ that is incoming to a leaf $x \in X$ of $N$  is present in $S'$, and so $a$ is  also an arc of $N'$ and so the leaf set of $N'$ contains $X$.  It remains to check that (1) $N'$ contains no other leaves, and (2) the only vertex of in-degree 0 in $N'$  is $\rho$.  
For (1), suppose $v$ is vertex of $N'$ that is not in $X$.  Then in $N$, $v$ has strictly positive out-degree.  If $v$ has out-degree 1 in $N$,  then the outgoing arc from $v$ is present in $S_1$ and thereby in $S$, while if $v$ has out-degree 2, at least one of the two outgoing arcs is present in $S$ by condition ($C_2$). Thus, $v$ cannot be a leaf of $N'$, establishing claim (1). 
Turning to claim (2), suppose that $v$ has in-degree 0. Then either (i)$'$ $v$ has in-degree 1 or 2 in $N$, or (ii)$'$  $v$ is the root vertex of $N$.
Now (i)$'$  cannot hold since the admissibility of $S$ implies that  at least one in-coming arc into $v$ is present in $N'$ (if $v$ has in-degree 1, then the incoming arc lies in $S_1$ and hence $S$, while if $v$ has in-degree 2, condition ($C_1$) implies that one incoming arc into $v$ is present in $N'$).  Case (ii)$'$   must now apply since every finite acyclic network has at least one arc of in-degree 0. This establishes claim (2), and thereby the ``if" direction in the first statement of Part (a). 

For the second statement of Part (a), we simply observe that the function $S \mapsto (V,S)$ from admissible subsets of $A$ to valid support trees for $N$ is a bijection since it has a (left and right) inverse in the opposite direction, namely
$T'= (V, S') \mapsto S'$.

{\em Part (b)}
We will show that any rooted phylogenetic  network can be translated directly into an instance of 2-SAT (a conjunction of clauses, each involving just two literals or their negations) in such a way that the existence of an admissible subset of arcs for $N$ corresponds to the satisfiability of the corresponding 2-SAT instance. 

Given $N$, let the set of literals be the arc set $A$, and consider the conjunction of the following clauses $C_\alpha$: 
$$C_N: =  \bigwedge_{a \in S_1} C_a \wedge  \bigwedge_{v\in V_{{\rm in}-2}} (C_v \wedge C'_v) \wedge  \bigwedge_{w \in V_{{\rm out}-2}} C_w,$$
where $V_{{\rm in}-2}$ (resp. $V_{{\rm out}-2}$)  is the set of vertices of $N$ of in-degree 2 (resp. out-degree 2) and for $a\in S_1$:
\begin{equation}
\label{eq1}
C_a = a;
\end{equation}
while for $v \in V_{{\rm in}-2}$ (with incoming arcs $a, b$):
\begin{equation}
\label{eq2}
C_v=(a \vee b) \mbox { and } C'_v = (\neg a \vee \neg b);
\end{equation}
and for $w \in V_{{\rm out}-2}$ (with outgoing arcs $a', b'$):
\begin{equation}
\label{eq3}
C_w = a' \vee b'.
\end{equation}

Notice that $C_N$ is an instance of 2-SAT (a conjunction of clauses, each of which involves just two literals or their negation),  and that if we interpret a truth assignment to $A$ as indicating whether $a\in A$ is an element of $S$ (`true') or of $A-S$ (`false')
then $C_N$ is satisfiable if and only if $N$ has an admissible subset $S$ (since the three types of clauses in (\ref{eq1})--(\ref{eq3}) respectively capture the conditions $S_1 \subseteq S$ and conditions ($C_1$) and ($C_2$) for admissibility). 

Part (b) now follows by the equivalence between admissibility and tree-based in part (a), and  the classic result (dating back to~\cite{kro}) that any instance of 2-SAT can be solved in polynomial time  (indeed in  linear time by more recent techniques \citep{asp}).
\end{proof}

\subsection{Proof of Proposition~\ref{mainpro}}

\begin{proof}
\mbox{ }
\noindent [(a) $\Leftrightarrow$ (b)] Suppose that $N$ is tree-based. Then for any tree-based representation for $N$, no two linking arcs can share the same vertex (by considering the various possible cases). Moreover, a linking arc is never incoming to an in-degree 1 vertex, or  outgoing from an out-degree 1 vertex, and so deleting linking arcs will not disconnect the network. Thus if we take $I$ to be the linking arcs in any tree-based representation for $N$ then $N'=(V, A-I)$ is an associated support tree for $N$.
Conversely, suppose that (b) holds for some set $I$. Since $N'$ is connected it has leaf set $X$, and so $N'$ is a subdivision of some rooted phylogenetic $X$-tree $T$. Now if we regard each arc in $I$ as a linking arc then we recover $N$ (since $I$ is independent, and $N'$ is connected, these arcs are all placed validly).

\bigskip

\noindent  [(a)$\Leftrightarrow$ (c)] If $N$ is tree-based, then any valid support tree $T'$ for $N$ satisfies the conditions specified in (c). 
Conversely suppose that $\tilde{T}$ is a rooted
spanning tree of $N$ (with root $\rho$) that contains the arcs in $S_1$, and has no leaves outside of $X$.  Since $\tilde{T}$ is a spanning tree it contains all vertices of $N$,
and any additional arcs in $N$ are either (i) from a vertex of out-degree 1 in $\tilde{T}$ to a vertex having in-degree and out-degree equal to  1  in $\tilde{T}$,  or
(ii) from a vertex of out-degree 0 in $\tilde{T}$ to a vertex having in-degree and out-degree equal to  1 in $\tilde{T}$; however,  case (ii)  is excluded by the assumption that $\tilde{T}$ has no leaves outside of $X$.   Thus if we let  $S$ be the set of arcs of $\tilde{T}$, then $S$ contains $S_1$, and  condition ($C_1$) holds (since $\tilde{T}$ is a tree), and condition ($C_2$) also holds from case (i). Thus $S$ is an admissible subset of arcs for $N$, and so $N$ is tree-based.
\end{proof}

\subsection{Proof of Proposition~\ref{mainpro2}}
\begin{proof}
Suppose $N$ is based on a tree. Then any antichain $\A$ of $N$ is also an antichain in any support tree $T'$ for $N$, since removing the linking arcs in returning to $T'$ from $N$ cannot create paths between vertices.
For each vertex $v \in \A$, select a leaf $x_v$ that lies below $v$ (i.e. there is a directed path from $v$ to $x_v$). Since $T'$ is a tree, these $|\A|$ are all arc-disjoint.  Moreover, the reinsertion of the linking arcs in moving from $T'$ to $N$ does not alter the  arc-disjointness of these $|\A|$ paths.

For the second claim, observe that $X$ is itself an antichain of $N$ of size $|X|$. Suppose there were an antichain $\A$ of $N$ of size strictly greater than $|X|$; we will show that this implies that $N$ is not tree-based.  Let  $\A_1$ and $\A_2$ be the sets of vertices in $\A$ that are leaves and non-leaves, respectively. Since $|\A|>|X|$ it follows that
 $|\A_2|\geq 1$. 
  Now if there were $|\A_2|$ arc-disjoint paths 
 from $\A_2$ to the leaves of $N$ then these $|\A_2|$ leaves together with $\A_1$ would comprise $|\A|$ distinct leaves. Since $|\A| >|X|$, this is not possible, and so
 $\A_2$ violates the antichain-to-leaf property, and hence $N$ is not a tree-based network. 
 \end{proof}

\subsection{Proof of Proposition \ref{prop1}}
\begin{proof}
Order the leaves of $T$ as $x_1, x_2, x_3, \ldots, x_n$.  For each $x_i$, consider the pendant arc $a_i$ of $T$ that is incident with $x_i$.
Place a linking arc from $a_i$ to $a_j$ for each  pair $i,j$ with $i< j$.  Above these $\binom{n}{2}$ arcs, place another set of $\binom{n}{2}$ linking arcs, again from  $a_i$ to $a_j$ for  each  pair $i,j$ with $i< j$. Continue this process so as to place a total of $n-1$ sets of $\binom{n}{2}$ such collections of linking arcs between the pendant arcs of $T$ to obtain a network $\widetilde{N}$ based on $T$ containing $\frac{1}{2}n(n-1)^2$ linking arcs altogether (see Fig. \ref{fig:displays-all-trees}).

\begin{figure}[ht]
\begin{center}
\includegraphics[width=8cm]{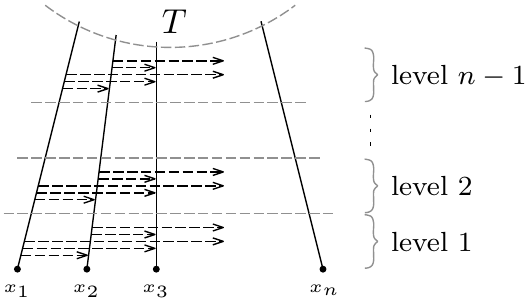}
\caption{A tree-based binary network that displays all rooted binary phylogenetic $X$-trees with $n$ leaves.}
\label{fig:displays-all-trees}
\end{center}
\end{figure}

We claim that $\widetilde{N}$ displays all rooted binary phylogenetic $X$-trees. To see this, note that any rooted binary phylogenetic $X$-tree $T'$ can be constructed by a sequence of 
$n-1$ steps of a `coalescent' process which starts with a graph of $n$ isolated leaves. At  each step this process joins two elements of the graph so-far constructed to a root vertex (the number of components of the resulting forest decreases by 1 at each step, and so we arrive at a tree after $n-1$ steps) -- for an example of this coalescent process, see Fig. 2.8 of \cite{sem}. The generous placement of the linking arcs in $\widetilde{N}$ allows for this coalescent process to be realised (for any tree $T'$) in $\widetilde{N}$. 
\end{proof}

\end{document}